\documentclass[%
 reprint,
 amsmath,amssymb,
 aps,
]{revtex4-1}

\usepackage{amsthm}
\usepackage{graphicx}
\usepackage{dcolumn}
\usepackage{bm}
\usepackage{amsthm}

\begin{document}

\preprint{APS/123-QED}

\title{Locally distinguish unextendible product bases by efficiently using entanglement resource}

\author{Zhi-Chao Zhang$^{1}$}
  \email{zhichao858@126.com}
  \author{Xia Wu$^{2}$}%
 \email{wuxia@cufe.edu.cn}
 \author{Xiangdong Zhang$^{1}$}%
 \email{zhangxd@bit.edu.cn}

\affiliation{%
 $^{1}$Beijing Key Laboratory of Nanophotonics and Ultrafine Optoelectronic Systems, School of Physics, Beijing Institute of Technology, Beijing, 100081, China\\
 $^{2}$School of Information, Central University of Finance and Economics, Beijing, 100081, China\\}

\date{\today}

\begin{abstract}
Any set of states which cannot be perfectly distinguished by local operations and classical communication (LOCC) alone, can always be locally distinguished using quantum teleportation with enough entanglement resource. However, in quantum information theory, entanglement is a very valuable resource, so it leaves the following open question: how to accomplish this task more efficiently than teleportation, that is, design the local discrimination protocol using less entanglement resource. In this paper, we first present two protocols to locally distinguish a set of unextendible product bases (UPB) in $5\otimes 5$ by using different entanglement resource. Then, we generalize the distinguishing methods for a class of UPB in $d\otimes d$, where $d$ is odd and $d\geqslant 3$. Furthermore, for a class of UPB in $d\otimes d$, where $d$ is even and $d\geqslant 4$, we prove that these states can also be distinguished by LOCC with multiple copies of low-dimensional entanglement resource. These results offer rather general insight into how to use entanglement resource more efficiently, and also reveal the phenomenon of less nonlocality with more entanglement.
\begin{description}
\item[PACS numbers]
03.67.Hk, 03.65.Ud, 03.67.Mn
\end{description}
\end{abstract}

\pacs{Valid PACS appear here}
\maketitle


\section{\label{sec:level1}Introduction\protect}
In quantum information theory, quantum information is always hidden in the quantum states. In addition, classical information can also be encoded into a set of orthogonal or nonorthogonal quantum states which are selected in advance. The decoding process can be seen as a protocol which can distinguish the encoded quantum states. Hence, it is fundamentally of interest to study quantum states discrimination problem. In this problem, a quantum system is prepared with a state which is secretly chosen from a known set, and the purpose is to determine in which state the system is. It is well known that a given set of orthogonal quantum states can be perfectly distinguished by performing a measurement on the entire system, and perfectly distinguishing nonorthogonal quantum states is not possible [1]. However, even if a given set only contains orthogonal multipartite quantum states, the situation may change dramatically when the physical conditions restrict our ability such that we are only allowed to do local operations and classical communication (LOCC), that is, the subsystems of a composite system are distributed among several spatially separated parties which are restricted to perform measurements only on their own subsystems and are allowed for any sequence of classical communication. Therefore, the local distinguishability of orthogonal quantum states has been widely studied [2-14], and practically applied in quantum secret sharing and quantum data hiding [15-18].

In the above study, product state may be very special because it admits local preparation according to some known rules, and therefore, it should be possible to learn about the state of the system with local measurements alone. However, Bennett \emph{et al.} presented a surprising result, that is, a complete orthogonal product bases in $3\otimes3$ cannot be distinguished by LOCC, and dubbed this phenomenon ``nonlocality without entanglement'' [2]. Then, various of related results have been presented [19-28]. There also exist ``incomplete bases'' to exhibit the property, known as unextendible product bases (UPB) [3,4], which means a set of mutually orthogonal product states satisfies the condition that no product state lies in the orthogonal complement of the subspace by these states, that is, the set of states cannot be extended by adding product state to it while preserving the orthogonality of the set. UPB cannot be distinguished perfectly by LOCC, and the projector onto that orthogonal complement is a mixed state which shows the fascinating phenomenon known as bound entanglement [3,29]. Thus, these states are of considerable interest in quantum information theory. 

For the locally indistinguishable orthogonal quantum states, the presence of additional entanglement can help us to change the distinguishability [30]. For example, using enough entanglement resource, we can teleport [31] the full multipartite state to a single party by LOCC, then this party can determine which state they were given [32]. However, in quantum information theory, entanglement is a very valuable resource, allowing remote parties to communicate in ways which were previously not thought possible, such as, the well-known protocols of teleportation [31], dense coding [33], data hiding [15,16]. In addition, entanglement can also be used to explore the potential power of quantum computer [34]. Therefore, an interesting question that remains to be answered is whether the above discrimination task can be accomplished more efficiently, i.e., using less entanglement resource.

In 2008, Cohen presented that certain classes of UPB in $m\otimes n (m\leq n)$ can be distinguished by LOCC with a $\lceil m/2 \rceil \otimes \lceil m/2 \rceil$ maximally entangled state, and left an interesting open question whether other UPB can also be locally distinguished by efficiently using entanglement resource [30]. Recently, for some locally indistinguishable orthogonal product states (not UPB), Zhang \emph{et al.} designed a discrimination protocol by using multiple copies of low-dimensional entanglement resource instead of a high-dimensional entanglement resource [35]. In addition, the method should be relatively easier to implement in real experiment because it only needs one equipment which can produce $2\otimes 2$ maximally entangled states instead of high-dimensional entangled states which will change for different sets of quantum states. However, whether the method is also applicative for UPB has not been answered [35]. 

Motivated by the above problems, in this paper, we first present a set of UPB that can be locally distinguished with a $3\otimes 3$ maximally entangled state in $5\otimes 5$. Then, we prove that this task can also be accomplished with the help of two copies of $2\otimes 2$ maximally entangled states. Furthermore, we generalize the above two results for a class of UPB in $d\otimes d$, where $d$ is odd and $d\geqslant 3$. Finally, for a class of UPB in $d\otimes d$, where $d=2n, n\geqslant 2$, which has been proved that are locally distinguishable with a $n\otimes n$ maximally entangled state \cite{S08}, here we show these states can also be perfectly distinguished by LOCC with $(n-1)$ copies of $2\otimes 2$ maximally entangled states. For the remaining open questions in [30,35], we have presented a positive answer. In addition, our results show that the locally indistinguishable quantum states may become distinguishable with a small amount of entanglement resource, and can also let people have a better understanding of the relationship between entanglement and nonlocality.

\theoremstyle{remark}
\newtheorem{definition}{\indent Definition}
\newtheorem{lemma}{\indent Lemma}
\newtheorem{theorem}{\indent Theorem}
\newtheorem{corollary}{\indent Corollary}

\def\QEDclosed{\mbox{\rule[0pt]{1.3ex}{1.3ex}}}
\def\QED{\QEDclosed}
\def\proof{\indent{\em Proof}.}
\def\endproof{\hspace*{\fill}~\QED\par\endtrivlist\unskip}

\section{Entanglement assisted discrimination}

In this section, we will present two different entanglement assisted discrimination protocols. First, we show a class of UPB in $5\otimes 5$ as follows, which has the structure of Fig.1 [36]. The state $|F\rangle$, known as the stopper state, is not shown, as it would cover the whole diagram. In addition, throughout the paper, we use the abbreviation $|i\pm j\pm \cdots \rangle=|i\rangle \pm |j\rangle \pm \cdots$, and ignore the normalization constant in many scenarios where the constant does not play any role.
\begin{figure}
\small
\centering
\includegraphics[scale=0.5]{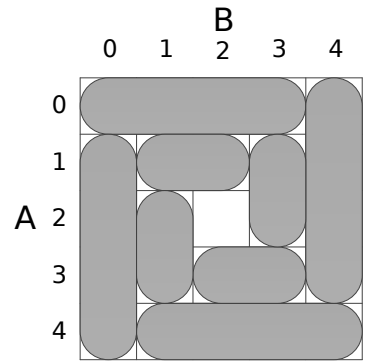}
\caption{Tile structure of a class of UPB in $5\otimes 5$. The labels on the rows (columns) correspond to Alice's (Bob's) standard basis, that is, $k$ is $|k\rangle$. The following tiles are the same.}
\end{figure}

\begin{eqnarray}
\label{eq.2}
\begin{split}
&|\phi_{1}\rangle=|0\rangle|0-1+2-3\rangle, |\phi_{2}\rangle=|0\rangle|0+1-2-3\rangle,\\
&|\phi_{3}\rangle=|0\rangle|0-1-2+3\rangle, |\phi_{4}\rangle=|0-1+2-3\rangle|4\rangle,\\
&|\phi_{5}\rangle=|0+1-2-3\rangle|4\rangle, |\phi_{6}\rangle=|0-1-2+3\rangle|4\rangle,\\
&|\phi_{7}\rangle=|4\rangle|1-2+3-4\rangle, |\phi_{8}\rangle=|4\rangle|1+2-3-4\rangle,\\
&|\phi_{9}\rangle=|4\rangle|1-2-3+4\rangle, |\phi_{10}\rangle=|1-2+3-4\rangle|0\rangle,\\
&|\phi_{11}\rangle=|1+2-3-4\rangle|0\rangle, |\phi_{12}\rangle=|1-2-3+4\rangle|0\rangle,\\
&|\phi_{13}\rangle=|1\rangle|1-2\rangle, |\phi_{14}\rangle=|1-2\rangle|3\rangle,\\
&|\phi_{15}\rangle=|3\rangle|2-3\rangle, |\phi_{16}\rangle=|2-3\rangle|1\rangle,\\
&|F\rangle=|0+1+2+3+4\rangle|0+1+2+3+4\rangle.
\end{split}
\end{eqnarray}

For the above UPB, the authors of [36] have presented that their nontrivial construction attributes a notable property compared to the trivial one which is always possible to distinguish few states perfectly from the UPB by orthogonality preserving LOCC. But in their case, not even a single state can be perfectly distinguished by such LOCC. This clearly indicates a stronger notion of local indistinguishability. Then, it is interesting whether the class of UPB needs more entanglement resource than [30]. In the following, we show locally distinguishing above quantum states only needs a $3\otimes 3$ maximally entangled state.

\begin{theorem}
In $5\otimes 5$, the states of (1) can be perfectly distinguished by LOCC with a $3\otimes 3$ maximally entangled state.
\end{theorem}

\begin{proof}
First of all, let Alice and Bob share a $3\otimes 3$ maximally entangled state $|\Psi\rangle_{ab}=|00\rangle+|11\rangle+|22\rangle$. 
Then, Alice performs a three-outcome measurement, each outcome corresponding to a rank-5 projector:

\begin{eqnarray}
\label{eq.2}
\begin{split}
A_{1}&=|00\rangle_{aA}\langle00|+|01\rangle_{aA}\langle01|+|02\rangle_{aA}\langle02|\\
&+|13\rangle_{aA}\langle13|+|24\rangle_{aA}\langle24|,\\
A_{2}&=|10\rangle_{aA}\langle10|+|11\rangle_{aA}\langle11|+|12\rangle_{aA}\langle12|\\
&+|23\rangle_{aA}\langle23|+|04\rangle_{aA}\langle04|,\\
A_{3}&=|20\rangle_{aA}\langle20|+|21\rangle_{aA}\langle21|+|22\rangle_{aA}\langle22|\\
&+|03\rangle_{aA}\langle03|+|14\rangle_{aA}\langle14|.
\end{split}
\end{eqnarray}

For operating with $A_{1}$ on systems $aA$, the picture of Fig. 2 is obtained, each of the initial states is transformed into:

\begin{figure}
\small
\centering
\includegraphics[scale=0.5]{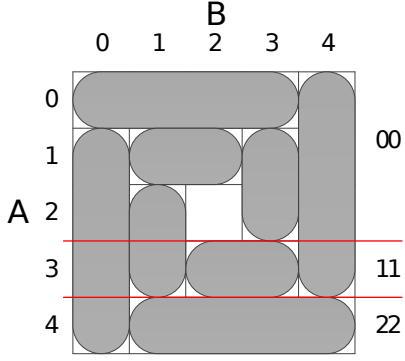}
\caption{Tile structure following Alice's first measurement with outcome $A_{1}$. The labels on the right rows correspond to Alice's and Bob's assisted systems, such as, $00$ is $|00\rangle_{ab}$. The following tiles are the same}
\end{figure}

\begin{eqnarray}
\label{eq.2}
\begin{split}
&|\phi'_{i}\rangle=|\phi_{i}\rangle|00\rangle,i=1, 2, 3, 13, 14,\\ 
&|\phi'_{i}\rangle=|\phi_{i}\rangle|11\rangle,i=15,\\
&|\phi'_{i}\rangle=|\phi_{i}\rangle|22\rangle,i=7, 8, 9,\\
&|\phi'_{4}\rangle=|0-1+2\rangle|4\rangle|00\rangle-|3\rangle|4\rangle|11\rangle,\\
&|\phi'_{5}\rangle=|0+1-2\rangle|4\rangle|00\rangle-|3\rangle|4\rangle|11\rangle,\\
&|\phi'_{6}\rangle=|0-1-2\rangle|4\rangle|00\rangle+|3\rangle|4\rangle|11\rangle,\\
&|\phi'_{10}\rangle=|1-2\rangle|0\rangle|00\rangle+|3\rangle|0\rangle|11\rangle-|4\rangle|0\rangle|22\rangle,\\
&|\phi'_{11}\rangle=|1+2\rangle|0\rangle|00\rangle-|3\rangle|0\rangle|11\rangle-|4\rangle|0\rangle|22\rangle,\\
&|\phi'_{12}\rangle=|1-2\rangle|0\rangle|00\rangle-|3\rangle|0\rangle|11\rangle+|4\rangle|0\rangle|22\rangle,\\
&|\phi'_{16}\rangle=|2\rangle|1\rangle|00\rangle-|3\rangle|1\rangle|11\rangle,\\
&|F\rangle\longrightarrow|0+1+2\rangle|0+1+2+3+4\rangle|00\rangle+\\
&|3\rangle|0+1+2+3+4\rangle|11\rangle+|4\rangle|0+1+2+3+4\rangle|22\rangle.
\end{split}
\end{eqnarray}

For operating with $A_{2}$ on systems $aA$, it creates new states which differ from the states (3) only by ancillary systems $|00\rangle_{ab}\longrightarrow|11\rangle_{ab}$, $|11\rangle_{ab}\longrightarrow|22\rangle_{ab}$, and $|22\rangle_{ab}\longrightarrow|00\rangle_{ab}$. Then, the latter can be handled using the exact same method as $A_{1}$. Operator $A_{3}$ is similar. Thus, we only need to discuss $A_{1}$.

Let us now describe how the parties can proceed from
here to distinguish these states. Bob makes a six-outcome
projective measurement, and we begin by considering the
first outcome, $B_{1}=|2\rangle_{b}\langle2|\otimes|1-2+3-4\rangle_{B}\langle1-2+3-4|$. The only remaining
possibility is $|\phi'_{7}\rangle$, which has thus been successfully identified. In the same way, Bob can identify $|\phi'_{8,9}, F\rangle$ by three projectors $B_{2}=|2\rangle_{b}\langle2|\otimes|1+2-3-4\rangle_{B}\langle1+2-3-4|$, $B_{3}=|2\rangle_{b}\langle2|\otimes|1-2-3+4\rangle_{B}\langle1-2-3+4|$, $B_{4}=|2\rangle_{b}\langle2|\otimes|1+2+3+4\rangle_{B}\langle1+2+3+4|$, respectively.

Using a rank-2 projector $B_{5}=|0\rangle_{b}\langle0|\otimes|4\rangle_{B}\langle4|+|1\rangle_{b}\langle1|\otimes|4\rangle_{B}\langle4|$ onto the Bob's Hilbert space, it leaves $|\phi'_{4,5,6}\rangle$, $|F\rangle\longrightarrow|0+1+2\rangle|4\rangle|00\rangle+|3\rangle|4\rangle|11\rangle$ to distinguish, and annihilates other states in (3). Then, Bob makes a projective measurement on system $b$ by projecting onto $|0+1\rangle_{b}$, $|0-1\rangle_{b}$ and $|2\rangle_{b}$. The latter has vanishing probability and the other two outcomes make Bob have the same state and Alice have orthogonal states. Thus, Alice can distinguish these states.

Bob's last outcome is a projector onto the remaining
part of Bob's Hilbert space $B_{6}=I-B_{1}-B_{2}-B_{3}-B_{4}-B_{5}$. This leaves $|\phi'_{1,2,3,10,\cdots,16}\rangle$ and $|F\rangle\longrightarrow|0+1+2\rangle|0+1+2+3\rangle|00\rangle+|3\rangle|0+1+2+3\rangle|11\rangle+|4\rangle|0\rangle|22\rangle$. Then, let Alice use the projector $A_{61}=|0\rangle_{a}\langle0|\otimes|0\rangle_{A}\langle0|$, it leaves $|\phi'_{1,2,3}\rangle$ and $|F\rangle\longrightarrow|0\rangle|0+1+2+3\rangle|00\rangle$, which can be easily distinguished by Bob. When Alice uses a projector $A_{62}=I-A_{61}$, it can leave $|\phi'_{10,\cdots,16}\rangle$, $|F\rangle\longrightarrow|1+2\rangle|0+1+2+3\rangle|00\rangle+|3\rangle|0+1+2+3\rangle|11\rangle+|4\rangle|0\rangle|22\rangle$ and annihilate other states. Then, Bob uses $B_{621}=(|0\rangle_{b}\langle0|+|1\rangle_{b}\langle1|+|2\rangle_{b}\langle2|)\otimes|0\rangle_{B}\langle0|$, it leaves $|\phi'_{10,11,12}\rangle$ and $|F\rangle\longrightarrow|1+2\rangle|0\rangle|00\rangle+|3\rangle|1\rangle|11\rangle+|4\rangle|0\rangle|22\rangle$. Similar to operator $B_{5}$, Bob can make a measurement in the Fourier basis on system $b$, and get the same state in Bob's party. Then, Alice can distinguish these states. Finally, for the projector $B_{622}=I-B_{621}$, it leaves $|\phi'_{13,\cdots,16}\rangle$ and $|F\rangle\longrightarrow|1+2\rangle|1+2+3\rangle|00\rangle+|3\rangle|1+2+3\rangle|11\rangle$, which can be perfectly distinguished by LOCC [30].

That is to say, we have succeeded in designing a protocol to perfectly distinguish the states (1) using LOCC with a $3\otimes 3$ maximally entangled state. 
This completes
the proof.          
\end{proof}

In the following, we present a different discrimination method for the states (1) with two copies of low-dimensional entanglement resource. 

\begin{theorem}
In $5\otimes 5$, the states of (1) can also be perfectly distinguished by LOCC with two copies of $2\otimes 2$ maximally entangled states.
\end{theorem}

\begin{proof}
Similarly, Alice and Bob first share two $2\otimes 2$ maximally entangled states $|\Psi\rangle_{a_{1}b_{1}}=|00\rangle+|11\rangle$ and $|\Psi\rangle_{a_{2}b_{2}}=|00\rangle+|11\rangle$. Then, Alice performs a two outcome measurement, each outcome corresponding to a rank-5 projector:

\begin{eqnarray}
\label{eq.2}
\begin{split}
A_{1}&=(|00\rangle_{a_{1}A}\langle00|+|01\rangle_{a_{1}A}\langle01|+|02\rangle_{a_{1}A}\langle02|\\
&+|13\rangle_{a_{1}A}\langle13|+|14\rangle_{a_{1}A}\langle14|)\otimes I_{a_{2}},\\
A_{2}&=(|10\rangle_{a_{1}A}\langle10|+|11\rangle_{a_{1}A}\langle11|+|12\rangle_{a_{1}A}\langle12|\\
&+|03\rangle_{a_{1}A}\langle03|+|04\rangle_{a_{1}A}\langle04|)\otimes I_{a_{2}}.
\end{split}
\end{eqnarray}

Similar to Theorem 1, we only need to consider $A_{1}$. Then, Alice performs a two outcome measurement, each outcome corresponding to a rank-5 projector:

\begin{eqnarray}
\label{eq.2}
\begin{split}
A_{11}&=(|00\rangle_{a_{2}A}\langle00|+|01\rangle_{a_{2}A}\langle01|+|02\rangle_{a_{2}A}\langle02|\\
&+|03\rangle_{a_{2}A}\langle03|+|14\rangle_{a_{2}A}\langle14|)\otimes I_{a_{1}},\\
A_{12}&=(|10\rangle_{a_{2}A}\langle10|+|11\rangle_{a_{2}A}\langle11|+|12\rangle_{a_{2}A}\langle12|\\
&+|13\rangle_{a_{2}A}\langle13|+|04\rangle_{a_{2}A}\langle04|)\otimes I_{a_{1}}.
\end{split}
\end{eqnarray}

In the same way, we only consider $A_{11}$ and the picture of Fig. 3 is obtained, that is, each of the initial states is transformed into:

\begin{figure}
\small
\centering
\includegraphics[scale=0.5]{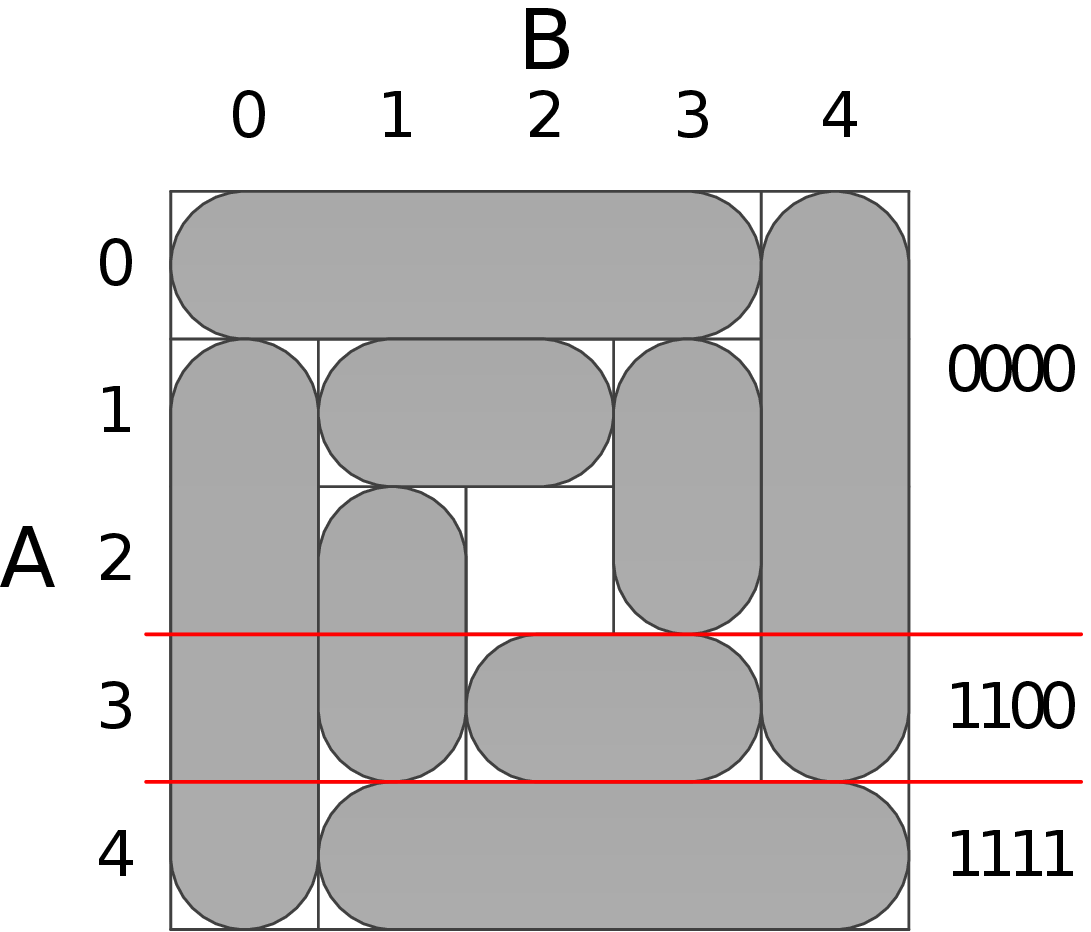}
\caption{Tile structure following Alice's first measurement with outcome $A_{11}$.}
\end{figure}

\begin{eqnarray}
\label{eq.2}
\begin{split}
&|\phi'_{i}\rangle=|\phi_{i}\rangle|0000\rangle,i=1, 2, 3, 13, 14,\\ 
&|\phi'_{i}\rangle=|\phi_{i}\rangle|1100\rangle,i=15,\\
&|\phi'_{i}\rangle=|\phi_{i}\rangle|1111\rangle,i=7, 8, 9,\\
&|\phi'_{4}\rangle=|0-1+2\rangle|4\rangle|0000\rangle-|3\rangle|4\rangle|1100\rangle,\\
&|\phi'_{5}\rangle=|0+1-2\rangle|4\rangle|0000\rangle-|3\rangle|4\rangle|1100\rangle,\\
&|\phi'_{6}\rangle=|0-1-2\rangle|4\rangle|0000\rangle+|3\rangle|4\rangle|1100\rangle,\\
&|\phi'_{10}\rangle=|1-2\rangle|0\rangle|0000\rangle+|3\rangle|0\rangle|1100\rangle-|4\rangle|0\rangle|1111\rangle,\\
&|\phi'_{11}\rangle=|1+2\rangle|0\rangle|0000\rangle-|3\rangle|0\rangle|1100\rangle-|4\rangle|0\rangle|1111\rangle,\\
&|\phi'_{12}\rangle=|1-2\rangle|0\rangle|0000\rangle-|3\rangle|0\rangle|1100\rangle+|4\rangle|0\rangle|1111\rangle,\\
&|\phi'_{16}\rangle=|2\rangle|1\rangle|0000\rangle-|3\rangle|1\rangle|1100\rangle,\\
&|F\rangle\longrightarrow|0+1+2\rangle|0+1+2+3+4\rangle|0000\rangle+\\
&|3\rangle|0+1+2+3+4\rangle|1100\rangle+|4\rangle|0+1+2+3+4\rangle|1111\rangle.
\end{split}
\end{eqnarray}
where $|0000\rangle$ means $|0000\rangle_{a_{1}b_{1}a_{2}b_{2}}$, others are similar.

Then, let Bob make a six-outcome projective measurement, and we begin by considering the
first outcome, $B_{1}=|1\rangle_{b_{1}}\langle1|\otimes|1\rangle_{b_{2}}\langle1|\otimes|1-2+3-4\rangle_{B}\langle1-2+3-4|$. The only remaining
possibility is $|\phi'_{7}\rangle$, which has thus been successfully identified. In the same way, Bob can identify $|\phi'_{8,9}, F\rangle$ by three projectors $B_{2}=|1\rangle_{b_{1}}\langle1|\otimes|1\rangle_{b_{2}}\langle1|\otimes|1+2-3-4\rangle_{B}\langle1+2-3-4|$, $B_{3}=|1\rangle_{b_{1}}\langle1|\otimes|1\rangle_{b_{2}}\langle1|\otimes|1-2-3+4\rangle_{B}\langle1-2-3+4|$, $B_{4}=|1\rangle_{b_{1}}\langle1|\otimes|1\rangle_{b_{2}}\langle1|\otimes|1+2+3+4\rangle_{B}\langle1+2+3+4|$, respectively.

Using a rank-2 projector $B_{5}=|0\rangle_{b_{1}}\langle0|\otimes|0\rangle_{b_{2}}\langle0|\otimes|4\rangle_{B}\langle4|+|1\rangle_{b_{1}}\langle1|\otimes|0\rangle_{b_{2}}\langle0|\otimes|4\rangle_{B}\langle4|$ onto the Bob's Hilbert space, it leaves $|\phi'_{4,5,6}\rangle$, $|F\rangle\longrightarrow|0+1+2\rangle|4\rangle|0000\rangle+|3\rangle|4\rangle|1100\rangle$ to distinguish, and annihilates other states in (6). Then, Bob make a projective measurement on system $b_{1}$ by projecting onto $|0+1\rangle_{b_{1}}$ and $|0-1\rangle_{b_{1}}$. The two outcomes make Bob have the same state and Alice have orthogonal states. Thus, Alice can distinguish these states.

Bob's last outcome is a projector onto the remaining
part of Bob's Hilbert space $B_{6}=I-B_{1}-B_{2}-B_{3}-B_{4}-B_{5}$. This leaves $|\phi'_{1,2,3,10,\cdots,16}\rangle$ and $|F\rangle\longrightarrow|0+1+2\rangle|0+1+2+3\rangle|0000\rangle+|3\rangle|0+1+2+3\rangle|1100\rangle+|4\rangle|0\rangle|1111\rangle$. Then, let Alice use the projector $A_{61}=|0\rangle_{a_{1}}\langle0|\otimes|0\rangle_{a_{2}}\langle0|\otimes|0\rangle_{A}\langle0|$, it leaves $|\phi'_{1,2,3}\rangle$ and $|F\rangle\longrightarrow|0\rangle|0+1+2+3\rangle|0000\rangle$, which can be easily distinguished by Bob. When Alice uses a projector $A_{62}=I-A_{61}$, it can leave $|\phi'_{10,\cdots,16}\rangle$, $|F\rangle\longrightarrow|1+2\rangle|0+1+2+3\rangle|0000\rangle+|3\rangle|0+1+2+3\rangle|1100\rangle+|4\rangle|0\rangle|1111\rangle$ and annihilate other states. Then, Bob uses $B_{621}=(|0\rangle_{b_{1}}\langle0|\otimes|0\rangle_{b_{2}}\langle0|+|1\rangle_{b_{1}}\langle1|\otimes|0\rangle_{b_{2}}\langle0|+|1\rangle_{b_{1}}\langle1|\otimes|1\rangle_{b_{2}}\langle1|)\otimes|0\rangle_{B}\langle0|$, it leaves $|\phi'_{10,11,12}\rangle$ and $|F\rangle\longrightarrow|1+2\rangle|0\rangle|0000\rangle+|3\rangle|1\rangle|1100\rangle+|4\rangle|0\rangle|1111\rangle$. Similar to operator $B_{5}$, Bob can make two measurements on systems $b_{1}$ and $b_{2}$ by the bases $|0\pm1\rangle$, and get the same state in Bob's party. Then, Alice can distinguish these states. Finally, for the projector $B_{622}=I-B_{621}$, it leaves $|\phi'_{13,\cdots,16}\rangle$ and $|F\rangle\longrightarrow|1+2\rangle|1+2+3\rangle|0000\rangle+|3\rangle|1+2+3\rangle|1100\rangle$, which can be perfectly distinguished by LOCC [35].

Therefore, we have also succeeded in designing a protocol to locally distinguish the states (1) with two copies of $2\otimes 2$ maximally entangled states. \end{proof}

In the following, we generalize the results for $d\otimes d$ quantum system, where $d$ is odd.

First, we show a class of UPB in $d\otimes d$, which has been constructed in [36] and has the tile structure given in Fig. 4, where $d$ is odd. In addition, $i$ means $\sqrt{-1}$ in the following.

\begin{figure}
\small
\centering
\includegraphics[scale=0.5]{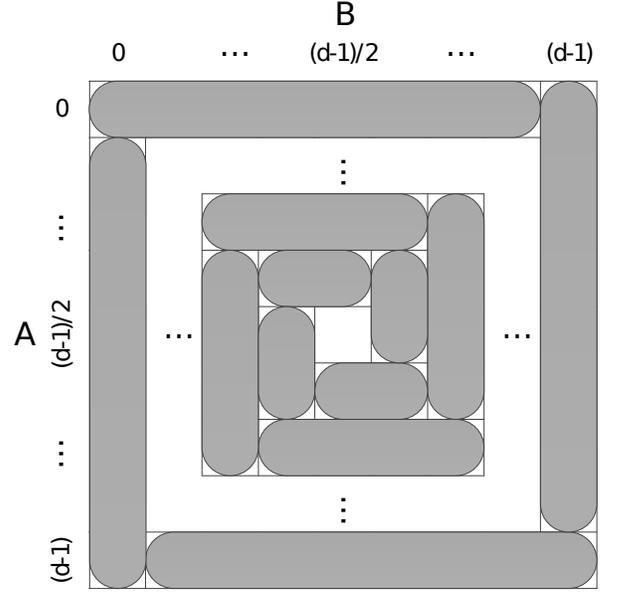}
\caption{Tile structure of a class of UPB in $d\otimes d$, $d$ is odd.}
\end{figure}

\begin{eqnarray}
\label{eq.2}
\begin{split}
&|\phi_{k}\rangle=|0\rangle|\sum\limits^{d-2}_{j=0}\omega^{jk}(j)\rangle,\omega=e^{\frac{2\pi i}{d-1}}, k=1,\cdots,d-2,\\
&|\phi_{k+d-2}\rangle=|\sum\limits^{d-2}_{j=0}\omega^{jk}(j)\rangle|d-1\rangle,\omega=e^{\frac{2\pi i}{d-1}}, k=1,\cdots,d-2,\\
&|\phi_{k+2d-4}\rangle=|d-1\rangle|\sum\limits^{d-1}_{j=1}\omega^{jk}(j)\rangle,\omega=e^{\frac{2\pi i}{d-1}}, k=1,\cdots,d-2,\\
&|\phi_{k+3d-6}\rangle=|\sum\limits^{d-1}_{j=1}\omega^{jk}(j)\rangle|0\rangle,\omega=e^{\frac{2\pi i}{d-1}}, k=1,\cdots,d-2,\\
&|\phi_{k+4d-8}\rangle=|1\rangle|\sum\limits^{d-3}_{j=1}\omega^{jk}(j)\rangle,\omega=e^{\frac{2\pi i}{d-3}}, k=1,\cdots,d-4,\\
&|\phi_{k+5d-12}\rangle=|\sum\limits^{d-3}_{j=1}\omega^{jk}(j)\rangle|d-2\rangle,\omega=e^{\frac{2\pi i}{d-3}}, k=1,\cdots,d-4,\\
&|\phi_{k+6d-16}\rangle=|d-2\rangle|\sum\limits^{d-2}_{j=2}\omega^{jk}(j)\rangle,\omega=e^{\frac{2\pi i}{d-3}}, k=1,\cdots,d-4,\\
&|\phi_{k+7d-20}\rangle=|\sum\limits^{d-2}_{j=2}\omega^{jk}(j)\rangle|1\rangle,\omega=e^{\frac{2\pi i}{d-3}}, k=1,\cdots,d-4,\\
&\vdots\\
&|\phi_{d^{2}-2d-2}\rangle=|(d-3)/2\rangle|(d-3)/2-(d-1)/2\rangle,\\
&|\phi_{d^{2}-2d-1}\rangle=|(d-3)/2-(d-1)/2\rangle|(d+1)/2\rangle,\\
&|\phi_{d^{2}-2d}\rangle=|(d+1)/2\rangle|(d-1)/2-(d+1)/2\rangle,\\
&|\phi_{d^{2}-2d+1}\rangle=|(d-1)/2-(d+1)/2\rangle|(d-3)/2\rangle,\\
&|F\rangle=|0+1+\cdots+(d-1)\rangle|0+1+\cdots+(d-1)\rangle.
\end{split}
\end{eqnarray}

Then, we present the local distinguishability of above quantum states with a $(d+1)/2\otimes (d+1)/2$ maximally entangled state.

\begin{theorem}
In $d\otimes d$, where $d$ is odd, the states of (7) can be perfectly distinguished by LOCC with a $(d+1)/2\otimes (d+1)/2$ maximally entangled state.
\end{theorem}

\begin{proof}
When $d = 3$, Cohen has exhibited that the states can be locally distinguished with a $2\otimes 2$ maximally entangled states \cite{S08}. 
When $d = 5$, we have presented it in Theorem 1. Next, we consider $d = n$, where n is odd. 
First, suppose when $n = d-2$, the above structure of $(d-3)^{2}+1$ states can be perfectly distinguished by LOCC with a $(d-1)/2\otimes (d-1)/2$ maximally entangled state. Then, we present that when $n=d$, the above structure of $(d-1)^{2}+1$ states can be locally distinguished with a $(d+1)/2\otimes (d+1)/2$ maximally entangled state.

First of all, Alice and Bob share a $(d+1)/2\otimes (d+1)/2$ maximally entangled state $|\Psi\rangle_{ab}=|00\rangle+|11\rangle+\cdots+|(d-1)/2(d-1)/2\rangle$. 
Then, Alice performs a $(d+1)/2$-outcome measurement, each outcome corresponding to a rank-$d$ projector:

\begin{eqnarray}
\label{eq.2}
\begin{split}
A_{1}&=|00\rangle_{aA}\langle00|+\cdots+|0(d-1)/2\rangle_{aA}\langle0(d-1)/2|\\
&+|1(d+1)/2\rangle_{aA}\langle1(d+1)/2|+\cdots\\
&+|(d-1)/2(d-1)\rangle_{aA}\langle(d-1)/2(d-1)|,\\
&\vdots\\
A_{(d+1)/2}&=|(d-1)/20\rangle_{aA}\langle(d-1)/20|+\cdots\\
&+|(d-1)/2(d-1)/2\rangle_{aA}\langle(d-1)/2(d-1)/2|\\
&+|0(d+1)/2\rangle_{aA}\langle0(d+1)/2|+\cdots\\
&+|(d-3)/2(d-1)\rangle_{aA}\langle(d-3)/2(d-1)|.\\
\end{split}
\end{eqnarray}

Similarly, we only need to consider $A_{1}$. For operating with $A_{1}$ on systems $aA$, the picture of Fig. 5 is obtained, each of the initial states is transformed into:

\begin{figure}
\small
\centering
\includegraphics[scale=0.5]{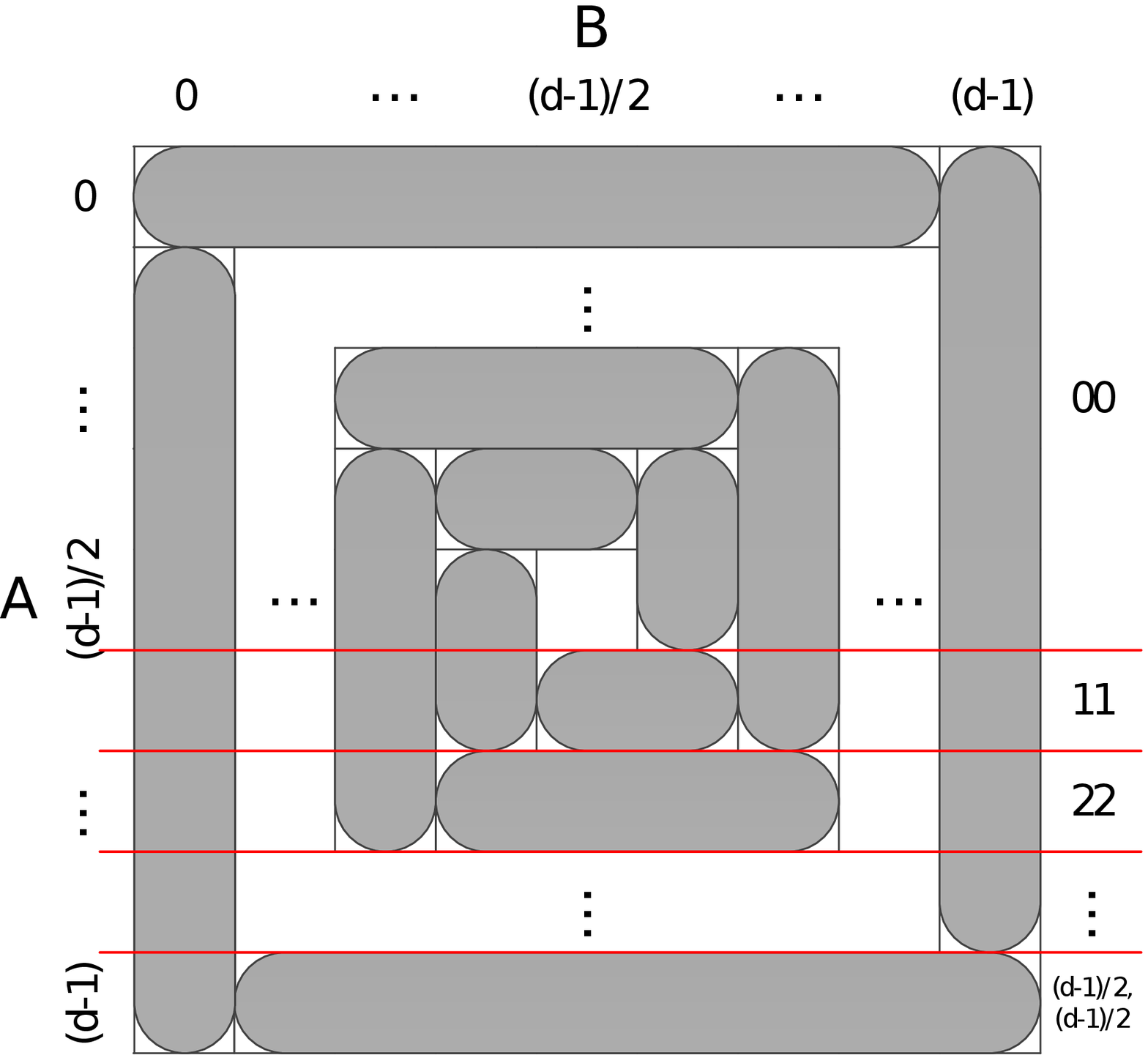}
\caption{Tile structure following Alice's first measurement with outcome $A_{1}$.}
\end{figure}

\begin{eqnarray}
\label{eq.2}
\begin{split}
|\phi'_{k}\rangle&=|\phi_{k}\rangle|00\rangle,k=1,\cdots,d-2,\\ 
|\phi'_{k+d-2}\rangle&=|\sum\limits^{(d-1)/2}_{j=0}\omega^{jk}(j)\rangle|d-1\rangle|00\rangle+\cdots+\\
&|\omega^{(d-2)k}(d-2)\rangle|d-1\rangle|(d-3)/2(d-3)/2\rangle,\\
&k=1,\cdots,d-2,\\
|\phi'_{k+2d-4}\rangle&=|\phi_{k+2d-4}\rangle|(d-1)/2(d-1)/2\rangle,\\
&k=1,\cdots,d-2,\\
|\phi'_{k+3d-6}\rangle&=|\sum\limits^{(d-1)/2}_{j=1}\omega^{jk}(j)\rangle|d-1\rangle|00\rangle+\cdots+\\
&|\omega^{(d-1)k}(d-1)\rangle|d-1\rangle|(d-1)/2(d-1)/2\rangle,\\
&k=1,\cdots,d-2,\\
&\vdots\\
|F\rangle\longrightarrow&|0+\cdots+(d-1)2\rangle|0+\cdots+(d-1)\rangle|00\rangle+\\
&\cdots+|(d-1)\rangle|0+\cdots+(d-1)\rangle\\
&|(d-1)/2(d-1)/2\rangle.
\end{split}
\end{eqnarray}

Similar to Theorem 1,  Bob can identify $|\phi'_{k}, F\rangle$ by projectors $B_{m}=|(d-1)/2\rangle_{b}\langle(d-1/)2|\otimes|\sum\limits^{d-2}_{j=0}\omega^{jm}(j)\rangle_{B}\langle\sum\limits^{d-2}_{j=0}\omega^{jm}(j)|$, $k=2d-3,\cdots,3d-6$, $m=1,\cdots,d-1$, respectively.

Using the projector $B_{d}=|0\rangle_{b}\langle0|\otimes|(d-1)\rangle_{B}\langle(d-1)|+\cdots+|(d-3)/2\rangle_{b}\langle(d-3)/2|\otimes|(d-1)\rangle_{B}\langle(d-1)|$ onto the Bob's Hilbert space, it leaves $|\phi'_{d-1,\cdots,2d-4}\rangle$, $|F\rangle\longrightarrow|0+\cdots+(d-1)2\rangle|(d-1)\rangle|00\rangle+\cdots+|(d-2)\rangle|(d-1)\rangle|(d-3)/2(d-3)/2\rangle$ to distinguish, and annihilates other states in (9). Next, Bob can make a measurement in the Fourier basis on system $b$, and get the same states in Bob's party. Then, Alice can distinguish these states.

Bob's last outcome is a projector onto the remaining
part of Bob's Hilbert space $B_{d+1}=I-B_{1}-\cdots-B_{d}$. This leaves $|\phi'_{1,\cdots,d-2,3d-5,\cdots,(d-1)^2}\rangle$ and $|F\rangle\longrightarrow|0+\cdots+(d-1)2\rangle|0+\cdots+(d-2)\rangle|00\rangle+\cdots+|(d-1)\rangle|0\rangle|(d-1)/2(d-1)/2\rangle$. Then, let Alice use the projector $A_{(d+1)1}=|0\rangle_{a}\langle0|\otimes|0\rangle_{A}\langle0|$, it leaves $|\phi'_{1,\cdots,d-2}\rangle$ and $|F\rangle\longrightarrow|0\rangle|0+\cdots+(d-2)\rangle|00\rangle$, which can be easily distinguished by Bob. When Alice uses a projector $A_{(d+1)2}=I-A_{(d+1)1}$, it can leave $|\phi'_{3d-5,\cdots,(d-1)^2}\rangle$, $|F\rangle\longrightarrow|1+\cdots+(d-1)2\rangle|0+\cdots+(d-2)\rangle|00\rangle+\cdots+|(d-1)\rangle|0\rangle|(d-1)/2(d-1)/2\rangle$ and annihilate other states. Then, Bob uses $B_{(d+1)21}=(|0\rangle_{b}\langle0|+\cdots+|(d-1)/2\rangle_{b}\langle(d-1)/2|)\otimes|0\rangle_{B}\langle0|$, it leaves $|\phi'_{3d-5,\cdots,4d-8}\rangle$ and $|F\rangle\longrightarrow|1+\cdots+(d-1)2\rangle|0\rangle|00\rangle+\cdots+|(d-1)\rangle|0\rangle|(d-1)/2(d-1)/2\rangle$. Similar to operator $B_{d}$, Bob can make a measurement in the Fourier basis on system $b$, and get the same states in Bob's party. Then, Alice has the orthogonal states which can be distinguished. Finally, for the projector $B_{(d+1)22}=I-B_{621}$, it leaves $|\phi'_{4d-7,\cdots,(d-1)^2}\rangle$ and $|F\rangle\longrightarrow|1+\cdots+(d-1)2\rangle|1+\cdots+(d-2)\rangle|00\rangle+\cdots+|(d-2)\rangle|1+\cdots+(d-2)\rangle|(d-3)/2(d-3)/2\rangle$. According to previous suppose, these states are locally distinguishable.

Thus, with a $(d+1)/2\otimes (d+1)/2$ maximally entangled state, these states of (7) can be perfectly distinguished by LOCC. The proof is completed.
\end{proof}

From the proof of Theorem 1 and 2, we can know that the two methods are very similar. Then, with $(d-1)/2$ copies of $2\otimes 2$ maximally entangled states, after Alice's measurement, one of outcomes can let these states of (7) transform into the structure of Fig. 6. Similar to Theorem 3, we only need to change assisted system $a$ into $a_{1} \cdots a_{(d-1)/2}$ and $b$ into $b_{1} \cdots b_{(d-1)/2}$ and easily prove that these states of (7) are also locally distinguishable, where $a_{1}b_{1}\cdots a_{(d-1)/2}b_{(d-1)/2}$ are the parties of assisted systems. Therefore, we have the following result.

\begin{theorem}
In $d\otimes d$, where $d$ is odd, the states of (7) can be locally distinguished with $(d-1)/2$ copies of $2\otimes 2$ maximally entangled states.
\end{theorem}

\begin{figure}
\small
\centering
\includegraphics[scale=0.5]{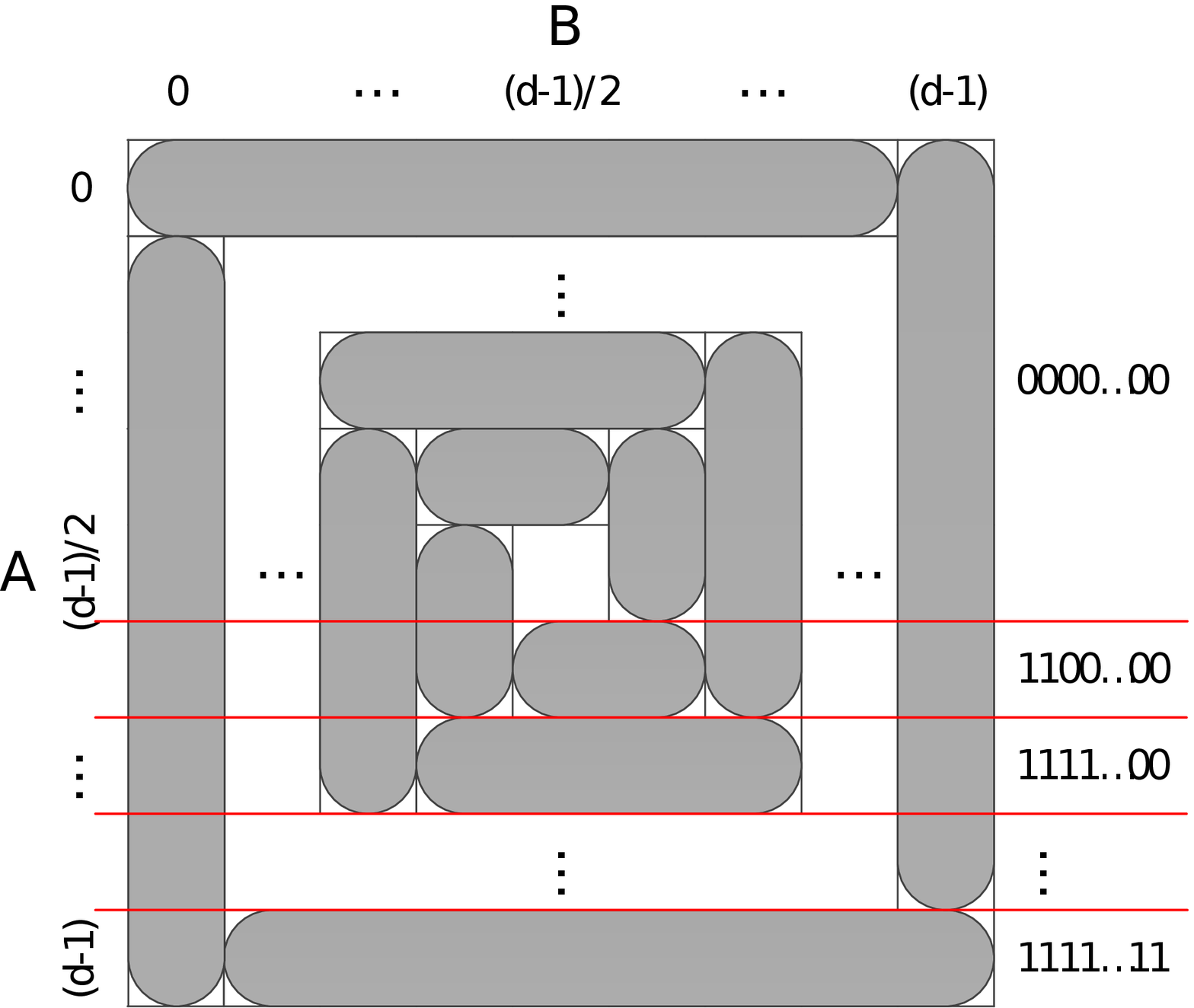}
\caption{Tile structure following Alice's measurement with outcome $A_{11\cdots1}$.}
\end{figure}

In addition, Cohen proved that a class of UPB in $d\otimes d, d=2n, n \geqslant 2$, which has the structure of Fig. 7, can be  locally distinguished with a $n\otimes n$ maximally entangled state [30]. In the same way, we can use $(n-1)$ copies of maximally entangled states to locally distinguish these states. First, Alice performs a series of measurements on systems $a_{1} \cdots a_{n-1}A$. Then, for one of outcomes, we get these states which has the structure of Fig. 8. Next, similar to the proof presented by Cohen in [30], we can easily get the following result.

\begin{theorem}
In $d\otimes d, d=2n, n \geqslant 2$, the states, which has the structure of Fig. 7, can be locally distinguished with $(n-1)$ copies of $2\otimes 2$ maximally entangled states.
\end{theorem}

\begin{figure}
\small
\centering
\includegraphics[scale=0.45]{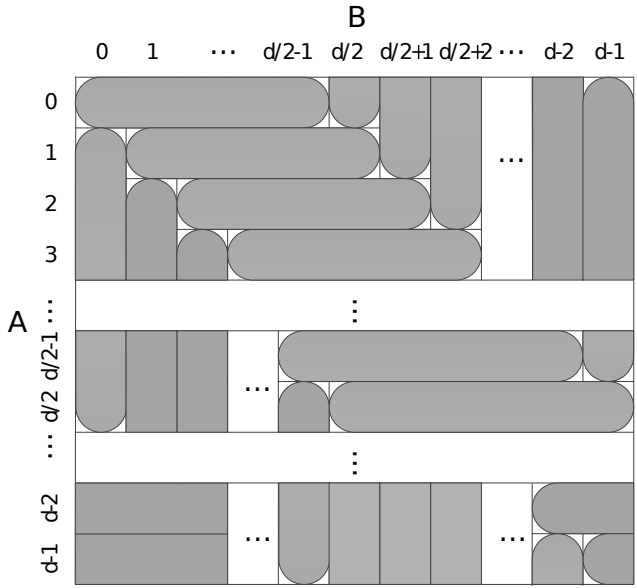}
\caption{Tile structure of a class of UPB in $d\otimes d$, $d$ is even \cite{DT03}. The set of vertical tile states are $|V_{mk}\rangle=|k\rangle|\sum\limits^{n-1}_{j=0}\omega^{jm}(j+k+1 \bmod d)\rangle$ and the set of horizontal tile states are $|H_{mk}\rangle=|\sum\limits^{n-1}_{j=0}\omega^{jm}(j+k\bmod d)\rangle|k\rangle$, $m=1,\cdots, n-1, k=0,\cdots,d-1$. In addition, there is also a stopper state $|F\rangle=\sum\limits^{d-1}_{i=0}\sum\limits^{d-1}_{j=0}|i\rangle|j\rangle$.}
\end{figure}

\begin{figure}
\small
\centering
\includegraphics[scale=0.45]{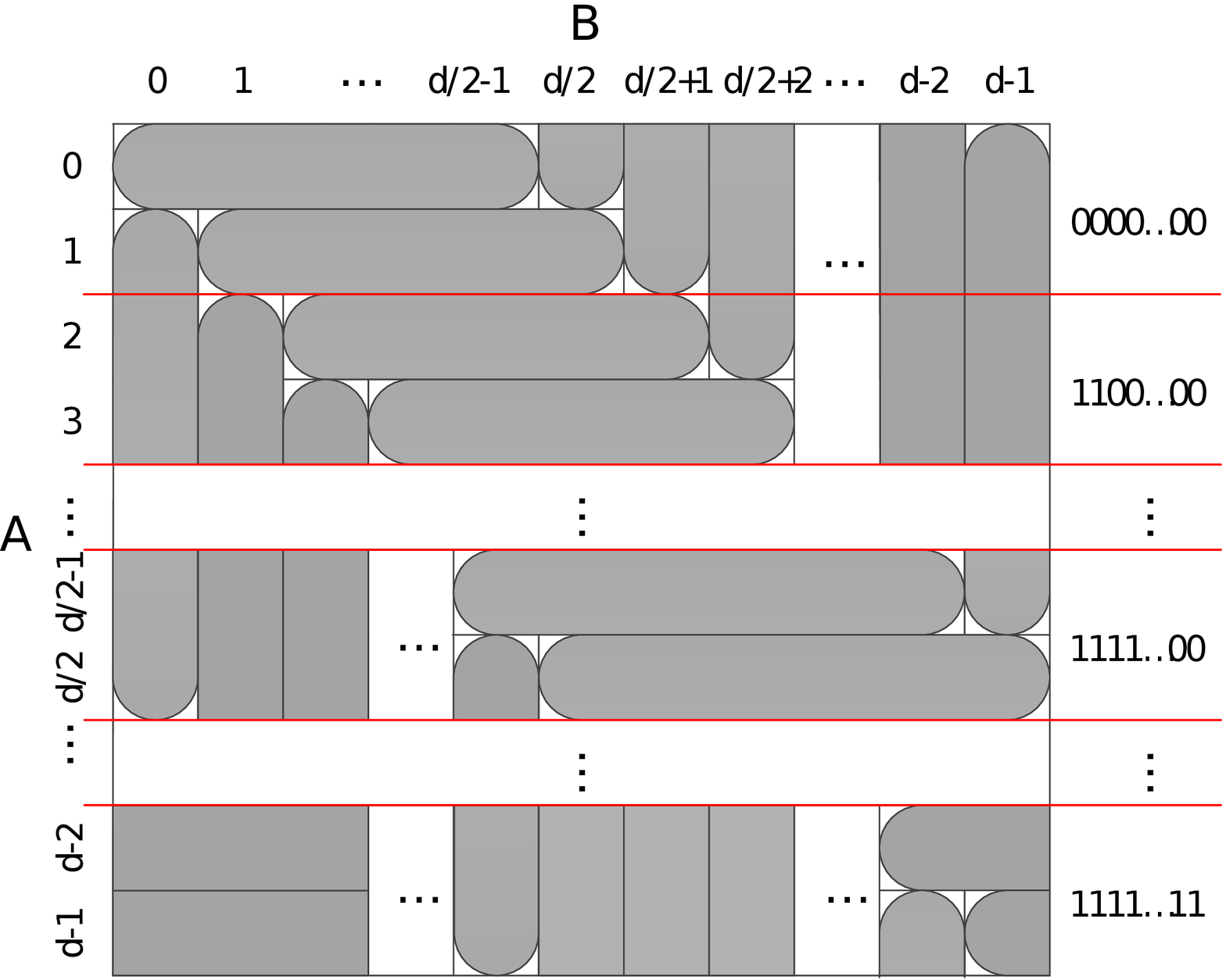}
\caption{Tile structure following Alice's measurement with outcome $A_{11\cdots1}$.}
\end{figure}

Different from previous a high-dimensional entanglement resource [30], we only use multiple copies of $2\otimes2$ maximally entangled states to complete this task. In addition, the method should be relatively easier to implement in real experiment because it only needs one equipment which can produce $2\otimes 2$ maximally entangled states instead of high-dimensional entangled states which will change for different sets of quantum states. 

\section{Conclusion}

In this paper, for entanglement as resource to distinguish unextendible product bases, we present two methods based on different entanglement resource. Our results can lead to a better understanding of the relationship between nonlocality and entanglement. Recently, there are also some results about orthogonal product states, but not UPB [37,38]. Finally, it is interesting what kind of state sets can always be perfectly distinguished by LOCC using multiple copies of $2\otimes 2$ maximally entangled states without teleportation.

\begin{acknowledgments}
The authors are grateful for the anonymous referees' suggestions to improve the quality of this paper. This work was supported by the Beijing Natural Science Foundation (Grant No. 4194088), the NSFC (Grants No. 11847210 and No. 61701553), the National Postdoctoral Program for Innovative Talent (Grant No. BX20180042), and the China Postdoctoral Science Foundation (Grant No. 2018M640070).
\end{acknowledgments}

\nocite{*}

\bibliography{apssamp}

\end{document}